\documentclass[journal]{IEEEtran}

\usepackage{cite}
\usepackage{amsmath,amssymb,amsfonts}
\usepackage{amsthm}
\usepackage{algorithmic}
\usepackage{graphicx}
\usepackage{textcomp}
\usepackage{booktabs}
\usepackage{subcaption}
\usepackage{multirow} 
\usepackage{xcolor}
\usepackage{soul}
\def\BibTeX{{\rm B\kern-.05em{\sc i\kern-.025em b}\kern-.08em
    T\kern-.1667em\lower.7ex\hbox{E}\kern-.125emX}}

\newtheoremstyle{italicProp}
  {3pt} 
  {3pt} 
  {\itshape} 
  {} 
  {\itshape\bfseries} 
  {.} 
  { } 
  {} 

\newtheorem{lemma}{Lemma}

\theoremstyle{definition}

\theoremstyle{remark}

\begin{document}

\title{Distributed Hierarchical Machine Learning for Joint Resource Allocation and Slice Selection in In-Network Edge Systems}

\author{Sulaiman Muhammad Rashid,
        Ibrahim Aliyu,
        Jaehyung Park,
        Seungmin Oh
        and Jinsul Kim%
\thanks{

This work was supported by the Institute of Information \& communications  Technology  Planning  \&  Evaluation(IITP) grant funded by the Korea government(MSIT), Project Name: Development of digital twin-based network failure prevention and operation management automation technology, Project Number: RS-2025-00345030, Contribution Rate: 50\%. This work was also supported by the Institute of Information \& Communications
Technology Planning \& Evaluation(IITP)-Innovative Human Resource
Development for Local Intellectualization program grant funded by the Korea government(MSIT)(IITP-2025-RS-2022-00156287), Contribution Rate: 50\%.

Sulaiman Muhammad Rashid, Ibrahim Aliyu, Jaehyung Park, Seungmin Oh and Jinsul Kim
are with the Department of Intelligent Electronics and Computer Engineering,
Chonnam National University, Gwangju 61186, South Korea
(e-mails: msulaimanrashid@jnu.ac.kr; aliyu@jnu.ac.kr;
hyeoung@jnu.ac.kr; osm5252kr@gmail.com; jsworld@jnu.ac.kr).%
}}
\maketitle

\begin{abstract}
The Metaverse promises immersive, real-time experiences; however, meeting its stringent latency and resource demands remains a major challenge. Conventional optimization techniques struggle to respond effectively under dynamic edge conditions and high user loads. In this study, we explore a slice-enabled in-network edge architecture that combines computing-in-the-network (COIN) with multi-access edge computing (MEC). In addition, we formulate the joint problem of wireless and computing resource management with optimal slice selection as a mixed-integer nonlinear program (MINLP). Because solving this model online is computationally intensive, we decompose it into three sub-problems—(SP1) intra-slice allocation, (SP2) inter-slice allocation, and (SP3) offloading decision—and train a distributed hierarchical DeepSets-based model (DeepSets-S) on optimal solutions obtained offline. In the proposed model, we design a slack-aware normalization mechanism for a shared encoder and task-specific decoders, ensuring permutation equivariance over variable-size wireless device (WD) sets. The learned system produces near-optimal allocations with low inference time and maintains permutation equivariance over variable-size device sets. Our experimental results show that DeepSets-S attains high tolerance-based accuracies on SP1/SP2 (Acc@1 = 95.26\% and 95.67\%) and improves multiclass offloading accuracy on SP3 (Acc = 0.7486; binary local/offload Acc = 0.8824). Compared to exact solvers, the proposed approach reduces the execution time by 86.1\%, while closely tracking the optimal system cost (within 6.1\% in representative regimes). Compared with baseline models, DeepSets-S consistently achieves higher cost ratios and better utilization across COIN/MEC resources.
\end{abstract}



\section{Introduction}
\label{sec:introduction}

The metaverse has attracted considerable attention from academia and industry in the last few years, and significant advancements have been in this technology owing to recent developments such as extended reality, 5G, and edge intelligence, among others \cite{10415393}. However, the extremely high demand for resources (for example, computing, networking, and storage) in the Metaverse is one of the major challenges that prevent Metaverse deployment \cite{9944868}. An enormous number of resources is required to fulfill the Quality-of-Service (QoS) requirements and meet user experience expectations in the Metaverse. Metaverse, which can simultaneously host a hundred thousand users, is expected to support millions of users concurrently \cite{bhattacharya2023towards}. 

To meet these demands, mobile edge computing (MEC) has been widely adopted as a key enabler, bridging the gap between computational resources end-users to reduce latency and improve service quality \cite{10754988}. Although MEC offers remedies through remote offloading (RO), it does not meet concurrent user demands \cite{aliyu2023dynamic}. “Computing in the network” (COIN), sometimes termed ”in-network computing,” has emerged as a complementary approach that extends computing resources closer to the user by leveraging distributed nodes within the network \cite{9919270}. COIN processes tasks at various points along the data path, reducing latency and improving resource utilization. The idea of COIN is not to replace the MEC, but to collaborate with the MEC \cite{aliyu2023dynamic}, leveraging additional computational resources distributed across the network. However, augmenting computing resources or enabling COIN leads to heightened power consumption \cite{10575680}. Segal et al. \cite{10336710} have shown the potential of in-network computing for reducing network utilization cost; however, effectively allocating COIN resources in real-time to adapt to continually shifting user demands, while ensuring overall system availability, poses a critical challenge.

Network slicing allows for the creation of virtual networks (slices) over a shared physical infrastructure, with each slice designed to meet the specific requirements of different applications \cite{8685766, 9295415}. In the context of the Metaverse, slicing can provide tailored computational, bandwidth, and latency conditions to accommodate diverse use cases (from gaming to healthcare) \cite{10158923}. However, the dynamic allocation of resources across these slices is a challenging task; dynamic allocation ensures that resources are optimally distributed between communication and computation while adapting to the varying demands of users and tasks.

Previous researchers \cite{sasan2024joint}, addressed joint network slicing and in-network computing resource allocation using a water-filling-based heuristic algorithm. However, their focus was solely on managing resources between slices (inter-slice), without considering the resource management issues within the slices (intra-slice). Furthermore, heuristic algorithms can achieve optimal solutions in a feasible amount of time, but they are not easy to derive \cite{lia2022}.

Recently, AI approaches as proposed in the concept of intelligent edge and considered a key 6G enabler have been used to determine how to allocate computing and communication resources across networks \cite{8884234}. This methodology holds promise by replacing complex mathematical modeling of the system to derive mathematically tractable heuristics with a data-driven understanding of the network. For example, the authors of \cite{araujo2021hybrid} focused on the allocation of network resources to provide service function chains (SFCs). After proposing a linear integer programming problem, they developed a hybrid strategy combining the optimal approach with supervised machine learning (ML) to find adequate network resources while decreasing the resolution time. 

In this paper, we address the joint network slicing, inter-slice \& intra-slice radio, and in-network resource management problems. We formulated the problem in \cite{10827552} as a mixed-integer non-linear programming problem (MINLP). The optimal solution achieved through a standard optimization solver is leveraged to train our DeepSet-S model by feeding it with information about the user's task requests and the state of the network. Therefore, the optimal solutions do not need to be computed online and instead are replaced with the model to provide quicker resource allocation and task placement decisions, albeit at the expense of heavy training procedures. DeepSets-S introduces a distributed hierarchical encoder–decoder architecture that learns to approximate solver-generated optimal policies for intra-slice, inter-slice, and offloading decisions. The model incorporates a shared encoder and task-specific decoders, ensuring permutation equivariance across variable-size wireless device (WD) sets and enabling scalable inference at both the slice and network levels. Specifically, our paper makes the following key contributions: 

\begin{itemize}
    \item We decomposed the problem into three sub-problems: we first derive the optimal intra-slice resource allocation policy, then address the optimal inter-slice resource allocation policies, and finally, find the optimal offloading decision vector.
    \item We developed a set-based encoder–decoder model that learns intra- and inter-slice allocation policies from solver data through a shared permutation-equivariant encoder and a novel slack-aware decoder that ensures feasible bandwidth and computing allocations within capacity limits.
    \item We employ a multi-head decoder to infer the optimal offloading decision vector. The architecture jointly predicts whether to execute locally or offload. 
    \item We performed an extensive evaluation to show that the performance of the resulting system significantly outperforms baseline resource allocation schemes.  
\end{itemize}

The rest of this paper is organized as follows. Section II reviews related works. Section III presents the system model. Section IV describes the optimization problem formulation. Section V details the proposed DeepSets-S model. Section VI reports experimental results and evaluations, Section VII presents the discussions and limitations, and Section VIII concludes the paper.

\section{Related works}

\subsection{Computing in the network (COIN) / In-Network Computing}
Sapio et al\cite{sapio2017network} describe COIN as a dumb idea whose time has come. It extends computing infrastructure closer to the user, typically within proximity to end-users or IoT devices \cite{kunze2021investigating}. Several studies have explored the potential of COIN in various domains, including metaverse, holographic applications, IoT, smart cities, and multimedia applications \cite{lia2022, 11050417, aliyu2023dynamic, 10575680, 10018567, 10660516}. The authors in \cite{lia2022} considered the problem of optimal placement of delay-constrained tasks to COIN nodes in an intelligent edge based on specific requirements to minimize network resource usage for low-latency task execution. Similarly, in a study conducted by \cite{11050417}, the authors developed a task-offloading solution for delay-constrained Metaverse systems to determine the optimal decision for offloading rendering tasks to COIN nodes or edge servers. The work in \cite{10018567} proposes an SDN-based COIN framework to place reusable, delay-sensitive computing tasks directly within the data forwarding path. Aliyu et al. \cite{10575680} propose a COIN-assisted edge architecture for Industrial IoT to reduce computation latency via partial offloading over URLLC links. The authors formulate a joint optimization of offloading decisions, ratios, and resource allocations, leveraging a hybrid game-theoretic approach combined with Double Deep Q-Networks (DDQN) to support low-complexity and distributed resource management. Wu et al. \cite{10660516} optimized COIN performance by adjusting the Computing Data Unit (CDU) size to reduce end-to-end latency in multi-hop in-network computing systems.

These works have not addressed, however, the potential impact of slicing on COIN nodes. So far, it is unclear how to perform joint management of communication and computing resources in a COIN-enabled system under network slicing.

\subsection{ML-based Resource Management}

Solving targeted optimization problems, such as minimizing energy consumption, latency, or bandwidth usage, either individually or jointly, typically involves addressing computationally intensive NP-hard problems. According to \cite{8847416}, machine learning (ML) techniques offer a promising alternative to traditional mathematical modeling of the system, enabling the derivation of mathematically tractable heuristics with a data-driven understanding of the system. For instance,  \cite{10041775} proposed a collaborative ML resource allocation scheme tailored for SDN-enabled fog computing, while \cite{lia2022} analyzed four classes of supervised learning methods; Decision tree (DT), Support vector machine (SVM), Multi-layer perceptron (MLP), and Bagged trees (BT) for task placement optimization at the network edge. In \cite{silva2024ml}, the authors proposed a predictive ML-based approach to monitor and guarantee the quality of network slice service by analyzing throughput and packet loss rate key quality indicators (KQI). Liu et al. \cite{10443273} proposed a learning-assisted end-to-end network slicing framework for resource allocation and task completion time minimization in heterogeneous networks. In \cite{hejazi2024learning}, SVM and MLP were employed for application placement in mobile edge computing servers. By formulating the problem as a two-stage stochastic optimization model, this study demonstrated the effectiveness of ML models in improving solution times and decision-making accuracy for user-to-server request allocation.

While previous works employing ML focus on instance based inference, where each decision is derived for individual users or requests in isolation, we addressed a set-structured/set input problem. We consider groups of WDs jointly, treating the task as a multiple instance learning problem. To capture this structure, we employ a DeepSet-based encoder–decoder architecture \cite{zaheer2017deep}, enabling the model to learn permutation-invariant representations of the WD sets.


\section{System Model}

We consider a slice-enabled network edge domain characterized by distinct application slices \( \mathcal{N} = \{1, 2, \ldots, N\} \), which have combinations of computing resources optimized for executing intensive tasks. We denote by \( \mathcal{I} = \{1, 2, \ldots, I\} \) as set of Wireless Devices (WD) that generate computationally intensive tasks each with varying computational demands. Task $i$ generated by $I$ characterized with input size $S_i$ can either be computed locally, or assigned to a slice $n$ through set \( \mathcal{A} = \{1, 2, \ldots, A\} \) of access points (AP).

In the case of assigning the task to slice, it goes through exactly one AP and can either be computed within the network on a set of COIN \( \mathcal{C} = \{1, 2, \ldots, C\} \), or can be offloaded to a set of MEC \( \mathcal{M} = \{1, 2, \ldots, M\} \). For simplicity, we denote $\overline{N}$ to be the set of all edge nodes (ENs), i.e  $\mathcal{C} \cup \mathcal{M}$. The APs and ENs form the set $\mathcal{E} \triangleq \overline{N} \cup \mathcal{A}$ of edge resorces. For all the task generated, there is expected number of instructions $R_i$ required to perform the computation. Since the WD, COIN and MEC have different instruction set of architectures, the required number of instructions may also vary. Therefore, for a task generated by WD $i$ we denote by $\mathcal{L}_i$, $\mathcal{L}_{i,n}$ as the expected number of instructions locally and in slice $n$ respectively which we considered them to be estimated by methods described in \cite{jovsilo2022joint,8314708}. 

We define set of decisions for task $i$ as \( \Delta_i = \{ i \} \cup \{ (a, j, n) \,|\, a \in \mathcal{A}, j \in \overline{N}, n \in \mathcal{N} \} \)  and we will use $\delta_i$ \(\in\) $\Delta_i$ to indicate the decision for WD $i$'s task i.e $\delta_i = i$ indicates that the task should be performed locally, and $\delta_i = (a, j, n)$ indicates that task should be offloaded through AP $a$ to node $j$ in slice $n$. Hence, we define a decision vector \( \boldsymbol(\delta) = (\delta_i)_{i \in \mathbb{I}} \) as the collection of the decisions of all WD's and we define the set \( \Delta = \times_{i \in \mathbb{I}} \Delta_i \), i.e., the set of all possible decision vectors:

WD's utilizing Access Point \( a \) within Slice \( n \) are captured in \( O_{a,n}(\boldsymbol{\delta}) \), representing \( i \in \mathcal{I} \) where \( \delta_i = (a, \cdot, n) \), spanning across all available access points \( a \) and slices \( n \). The aggregation of these sets across all slices yields \( O_a(\boldsymbol{\delta}) \), indicating all WDs employing Access Point \( a \) for task processing. Similarly, WD's employing node \( j \) within Slice \( n \) are compiled in \( O_{j,n}(\boldsymbol{\delta}) \), denoting \( i \in \mathcal{I} \) where \( \delta_i = (\cdot, j, n) \), with \( O_j(\boldsymbol{\delta}) \) encompassing all WD's utilizing node \( j \) across slices. The singleton set \( O_i(\boldsymbol{\delta}) \) distinguishes whether a WD \( i \) performs local computation (\( O_i(\boldsymbol{\delta}) = \{ i \} \)) or not (\( O_i(\boldsymbol{\delta}) = \emptyset \)). Finally, \( O_l(\boldsymbol{\delta}) \) represents all WDs performing local computation within the given decision vector \( \boldsymbol{\delta} \).

Fig. \ref{Architecture} shows an example of a slicing-enabled COIN-MEC architecture with $\mathcal{N} = 3$ slices, $\mathcal{I} = 6$ WDs, $\mathcal{A} = 3$ APs, $\mathcal{M} = 1$ MEC and $\mathcal{C} = 6$ COINs.

\subsection{Communication Resources}

In the network, the communication resources are managed both at the network level and at the slice level. At the network level, the radio resources of each access point (AP) are shared across the slices based on an interslice radio resource allocation policy, denoted as \( \mathbb{R}_{\omega}: \Delta \rightarrow \mathbb{R}_{[0,1]}^{|\mathcal{A}|\times|\mathcal{N}|} \). This policy determines the inter-slice radio resource provisioning coefficients \( \omega_{n}^{a} \), where \( \forall (a, n) \in \mathcal{A} \times \mathcal{N} \), \( \omega_{n}^{a} \leq 1 \).

At the slice level, the radio resources assigned to each slice are shared among the WD's according to an intraslice radio resource allocation policy \( \mathbb{R}^{n}_{\phi_a}: \Delta \rightarrow \mathbb{R}^{|\mathcal{A}|\times|\mathcal{I}|}_{[0,1]} \). This policy determines the intra-slice radio resource provisioning coefficients \( \phi_{i,a}^{n} \in [0,1] \), where \( \forall a \in \mathcal{A} \) and \( \forall i \in O_{a,n}(\boldsymbol{\delta}) \), such that \(\sum_{i \in O_{a,n}(\boldsymbol{\delta})} \phi_{i,a}^{n} \leq 1 \), \(\forall (a, n) \in \mathcal{A} \times \mathcal{N} \).

The achievable Physical rate of a WD \( i \) at AP \( a \), denoted as \( R_{i,a} \), captures the expected channel conditions, which can be estimated through historical measurements. Given \( R_{i,a} \) and the provisioning coefficients \( \omega_{a}^{n} \) and \( \phi_{i,a}^{n} \), the uplink rate of WD \( i \) at AP \( a \) in slice \( n \) is expressed as:

\begin{equation}
\mathcal{U}_{i,a}^{n}(\boldsymbol{\delta}, \mathbb{R}_b, \mathbb{R}_{\phi_a}^{n}) = \omega_{a}^{n} \phi_{i,a}^{n} R_{i,a}
\end{equation}

The uplink rate, along with the input data size \( S_i \), determines the transmission time of WD \( i \) in slice \( n \) at AP \( a \):

\begin{equation}
T^{tx,n}_{i,a}(\boldsymbol{\delta}, \mathbb{R}_b, \mathbb{R}_{\phi_a}^{n}) = \frac{S_i}{\mathcal{U}_{i,a}^{n}(\boldsymbol{\delta}, \mathbb{R}_b, \mathbb{R}_{\phi_a}^{n})}
\end{equation}


\begin{figure}[!ht]
  \hspace*{-0.11\linewidth}
  \includegraphics[width= 1.17\linewidth]{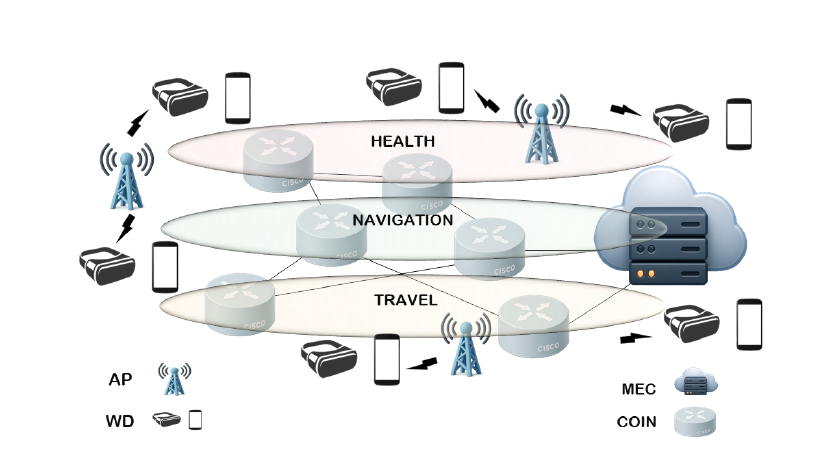}
  \caption{Slicing enabled COIN-MEC System with $\mathcal{N} = 3$ slices, $\mathcal{I} = 6$ WDs, $\mathcal{A} = 3$ APs, $\mathcal{M} = 1$ MEC and $\mathcal{C} = 6$ COINs}
  \label{Architecture}
\end{figure}

\subsection{Computing Resources}

In our system model, we distinguish between two main categories of computing resources: edge resources and local resources. Within this framework, each slice \( n \) is equipped with a specific combination of computing resources $\overline{N}$ equal to $\mathcal{C} \cup \mathcal{M}$ tailored for executing particular type of task, encompassing CPUs, GPUs, NPUs, and FPGAs. We denote by \( F_{j}^{n} \) the computing capability of node $j$ in slice $n$.   The allocation of these computing resources among the WD's is controlled by intra-slice computing power allocation policies \( \mathbb{R}^{n}_{\phi_{j}} : \Delta \rightarrow \mathbb{R}^{|\mathcal{\overline{N}}|\times|I|}_{[0,1]}\). This policy determines the provisioning coefficients \( \phi_{i,j}^{n} \in [0,1], \quad \forall i \in  O_{j,n}(\boldsymbol{\delta}) \), such that $\sum_{i \in  O_{j,n}(\boldsymbol{\delta})}\phi_{i,j}^{n} \leq 1, \quad \forall (j,n) \in \mathcal{N} \times \overline{N}$ ensuring that each WD receives a portion of the available computing power within the slice. Specifically, the computing capability $F_{j}^{n}$ allocated to WD \( i \) within the  cloud in slice \( n \) is expressed in the following equation: 
\begin{equation}
\label{eqn3}
F_{i,j}^{n}(\boldsymbol{\delta}, \mathbb{R}^{n}_{\phi_j}) = \phi_{i,j}^{n} F_{j}^{n} \qquad  j \in \overline{N}
\end{equation}

We assumed like in the literature \cite{11050417, lia2022}, and justified by empirical studies \cite{zukerman2013introduction}, the arrival rate of requests of task generated by WD $i$ at node $j$ of slice $n$ follows a Poisson distribution with parameter $\alpha_i$. The expected number of instructions \( \mathcal{L}_{i,n} \) required to execute a task generated by WD $i$ in slice $n$ is exponentially distributed. Therefore, all nodes can form an M/M/1 queuing model to process their corresponding computing tasks. The total request arrival rate at node $j$, which still follows a Poisson distribution, can be defined as:

\begin{equation}
\label{eqn5}
\sum_{i \in \mathcal{I}} \delta_{ij}\alpha_i, \qquad  j \in \overline{N}
\end{equation}

where $\delta_{ij}$ is the binary decision variable. It is equal to 1 if task generated by WD $i$ is executed by node $j$ and equal to 0 otherwise.

Thus, exploiting the queuing theory, the average computation time for a generic task at node $j$ of slice $n$ can be derived as:

\begin{equation}
T^{ex}_{i,j}(\boldsymbol{\delta}, \mathbb{R}^{n}_{\phi_{j}}) = \frac{1}{\frac{F_i,j^n}{\mathcal{L}_{i,n}} - \sum_{i \in \mathcal{I}} \delta_{ij}\alpha_i}, \qquad  j \in \overline{N}
\end{equation}

To keep the queue stable, the average arrival rate, $\alpha_i$ should be smaller than the average service rate, as described by the following equation:

\begin{equation}
\frac{F_i,j^n}{\mathcal{L}_{i,n}} - \sum_{i \in \mathcal{I}} \delta_{ij}\alpha_i > 0, \qquad  j \in \overline{N} 
\end{equation}

In addition to the cloud resources, each WD possesses local computing capabilities denoted by \( F_{i}^{l} \) which may vary over different devices. The local execution latency \( T^{ex}_{i} \) of WD \( i \) is simply expressed as 
\begin{equation}
T^{ex}_{i} = \frac{\mathcal{L}_i}{F^{l}_{i}}
\end{equation}

\subsection{Cost Model}
We define the system cost as the aggregate completion time of all WD's. Before providing a formal definition, we introduce the shorthand notation:

\begin{equation}
\mathcal{T}_{i,e}^{n} = \begin{cases}
\frac{S_i}{R_{i,e}} & \text{if } i \in \mathcal{I}, e \in \mathcal{E} \cap A \\
\frac{1}{\frac{F_i,j^n}{\mathcal{L}_{i,n}} - \sum_{i \in \mathcal{I}} \delta_{ij}\alpha_i} & \text{if } i \in \mathcal{I}, e \in \mathcal{E} \cap \overline{N} \\
\end{cases}
\end{equation}

For $e \in \mathcal{E} \cap A$ we have that $e$ is a communication resource and $\mathcal{T}_{i,e}^{n}$ is the minimum transmission time that WD $i$ will achieve if it is the only WD offloading its computation through AP $e$ in slice $n$. Similarly, for $e \in \mathcal{E} \cap C$ we have that $e$ is a computing resource and $\mathcal{T}_{i,e}^{n}$ is the minimum execution time that WD $i$ will achieve if it is the only WD offloading its computation to Node $j$ in slice $n$

To facilitate notation, we define the indicator function for WD \( i \):

\begin{equation}
{I}(\delta_i, \delta) = \begin{cases}
1 & \text{if } \delta_i = \delta \\
0 & \text{otherwise}
\end{cases}
\end{equation}

The cost of WD \( i \) is then determined by the task completion time, considering both local computation and offloading to edge:

\begin{equation}
\begin{split}
C_i(\boldsymbol{\delta}, \mathbb{R}_\omega, \mathbb{R}_{\phi_a}, \mathbb{R}_{\phi_j}) = & T^{\text{ex}}_{i} {I}(\delta_i, i) \\
& + \sum_{n \in N} \sum_{j \in \overline{N}} \sum_{a \in A} ( \frac{\mathcal{T}^{n}_{i,a}}{\omega_{a}^{n} \phi_{i,a}^{n}}  \\
& \left. + \frac{\mathcal{T}^{n}_{i,j}}{\phi_{i,j}^{n}}  \right){I}(\delta_i, (a, j, n))
\end{split}
\end{equation}

Here, \( (\mathbb{R}_{\phi_a}, \mathbb{R}_{\phi_j}) = ((\mathbb{R}_{\phi_a}^{n}, \mathbb{R}_{\phi_j}^{n},)) n \in N \)  represents the collection of slice policies.

The cost of each slice \( n \) is calculated as the sum of transmission and execution times for all WD's offloading their tasks in slice \( n \):

\begin{equation}
C_{(n)}(\boldsymbol{\delta}, \mathbb{R}_\omega, \mathbb{R}^n_{\phi_a}, \mathbb{R}^n_{\phi_j}) = \sum_{e \in \mathcal{E}} \sum_{i \in O_{e,s}(\boldsymbol{\delta})} \frac{\mathcal{T}_{i,e}^n}{\omega_e^n \phi_{i,e}^{n}} 
\end{equation}

Here, \( \omega_e^n = 1 \) if \( e \) is a computing resource.

Finally, the system cost is expressed as the sum of individual WD costs and slice costs:

\begin{equation}
\begin{split}
C(\boldsymbol{\delta}, \mathbb{R}_\omega, \mathbb{R}_{\phi_a}, \mathbb{R}_{\phi_j}) = & \sum_{i \in I} C_i(\boldsymbol{\delta}, \mathbb{R}_\omega, \mathbb{R}^n_{\phi_a}, \mathbb{R}^n_{\phi_j}) \\
& + \sum_{n \in N} C_{(n)}(\boldsymbol{\delta}, \mathbb{R}_b, \mathbb{R}^n_{\phi_a}, \mathbb{R}^n_{\phi_j})
\end{split}
\end{equation}

This system cost formulation accounts for both local and offloaded computation across slices.

\section{Optimization Problem Formulation}

Our goal is to minimize the system cost by finding the optimal decision from vector $\textbf{d}$ of offloading decisions. From the above analytical results, the problem can be expressed mathematically as a mixed-integer non linear programming (MINLP) problem and is formulated as follows:

\begin{equation}
\label{Optimal equation}
\begin{aligned}
& \underset{\boldsymbol{\delta}, \mathbb{R}_\omega, \mathbb{R}_{\phi_a} , \mathbb{R}_{\phi_j}}{\text{min}} \quad C(\boldsymbol{\delta}, \mathbb{R}_\omega, \mathbb{R}_{\phi_a} , \mathbb{R}_{\phi_j} ) 
\end{aligned}
\end{equation}

s.t.

\begin{align}
\text{} & \quad \sum_{\delta \in \Delta_i} {I}(\delta_i, \delta)=1, \quad i \in \mathcal{I} \tag{13a} \label{eq:C1} \\
\text{} & \quad T^{ex}_{i,j}(\boldsymbol{\delta}, \mathbb{R}^{n}_{\phi_{j}}) \leq T^{ex}_{i}, \quad i \in \mathcal{I} \tag{13b} \label{eq:C2} \\
\text{} & \quad \frac{F_i,j^n}{\mathcal{L}_{i,n}} - \sum_{i \in \mathcal{I}} \delta_{ij}\alpha_i > 0, \quad  j \in \overline{N}  \tag{13c} \label{eq:C3} \\
\text{} & \quad \sum_{n \in N} \omega_{a}^{n} \leq 1, \quad a \in \mathcal{A} \tag{13d} \label{eq:C4} \\
\text{} & \quad \sum_{i \in  O_{e,n}(\boldsymbol{\delta})}\phi_{i,e}^{n} \leq 1, \quad n \in \mathcal{N}, e \in \mathcal{E} \tag{13e}
\label{eq:C5} \\
\text{} & \quad \omega_{a}^{n} \geq 0, \quad a \in \mathcal{A} \tag{13f}
\label{eq:C6} \\
\text{} & \quad \phi_{i,e}^{n} \geq 0, \quad n \in \mathcal{N}, e \in \mathcal{E} \tag{13g} \label{eq:C7} 
\end{align}

The constraint \ref{eq:C1} ensures that each WD performs a computation locally or offloads its task to exactly one logical resource on the edge $(a, j, n) \,|\, a \in \mathcal{A}, j \in \overline{N}, n \in \mathcal{N}$. The constraint \ref{eq:C2} indicates that the task completion time when offloading is not greater than when computing locally. The constraint \ref{eq:C3} forces the average service rate of the edge nodes to be higher than the average task arrival rate in the case of offloading. Constraints \ref{eq:C4} \& \ref{eq:C5} enforce limitations on the amount of radio resources that can be provided to an AP in each slice, and the amount of computing resources of an edge node that can be provided to each WD in each slice. Constraints \ref{eq:C6} \& \ref{eq:C7} are non-negativity constraints that ensure the values of the provisioning coefficients are non-negative.

\section{Distributed Hierarchical DeepSet-S model for resource management and task offloading}\label{ML-model}

The MINLP problem formulated in equation \ref{Optimal equation} corresponds to the Generalized Assignment Problem (GAP) \cite{cattrysse1992survey}, which is an NP-hard problem in the combinatorial optimization literature. Due to its complexity, resolving it by using a standard optimization solver is not feasible; it would take a long time to arrive at the optimal solution. Therefore, a more dynamic solution is needed to solve the MINLP problem efficiently, even at the expense of a suboptimal solution.

For this purpose, consistent with the recent studies \cite{11050417,lia2022}, we decompose the MINLP into a sequence of $\mathcal{N}+2$ optimization subproblems, which are solved sequentially using machine learning (ML)-based techniques.

\begin{figure}[!ht]
  \centering
  \includegraphics[width= 1.07\linewidth]{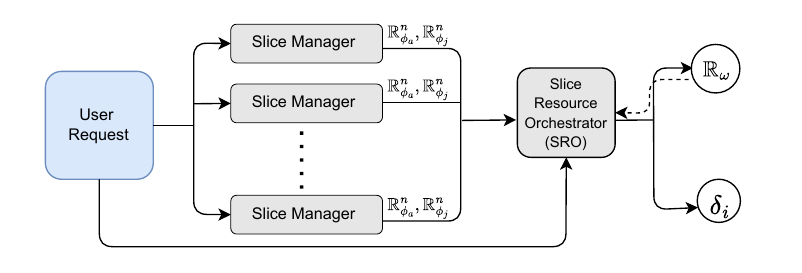}
  \caption{Distributed Hierarchical AI Model}
  \label{hierarchical}
\end{figure}

\textbf{Subproblem 1 (SP1)}:
As a first step in the decomposition, consider the problem of finding an optimal collection $\mathbb{R}_{\phi_a}^n, \mathbb{R}_{\phi_j}^n$ of slice resource allocation policies for fixed offloading decision vector $\boldsymbol{\delta}$ and interslice policy $\mathbb{R}_\omega$, for which constraint \ref{eq:C2} can be satisfied. Then the solution to the problem is given by:

\begin{equation}
\min_{\mathbb{R}_{\phi_a} , \mathbb{R}_{\phi_j}} \sum_{n \in \mathcal{N}} \sum_{e \in \mathcal{E}} \sum_{i \in O_{e,n}(\delta)} \frac{\mathcal{T}_{i,e}^n}{\phi_{i,e}^n} 
\end{equation}
\begin{align}
\text{} & \quad \textbf{st.} (\ref{eq:C2}),(\ref{eq:C3}), (\ref{eq:C5}), (\ref{eq:C7}) \tag{14a} \label{eq:C14a} 
\end{align}

\begin{figure}[!ht]
  \centering
  \includegraphics[width= 1\linewidth]{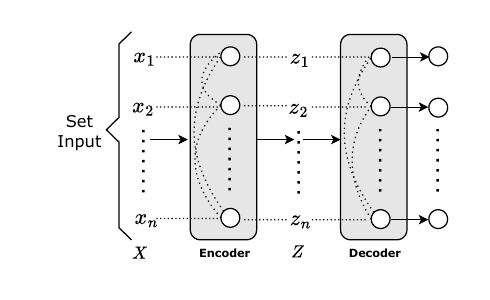}
  \caption{DeepSets-S Model}
  \label{SP1_transformer}
\end{figure}


\textbf{Subproblem 2 (SP2)}:
As a second step in the decomposition, consider the problem of finding an optimal inter-slice radio resource allocation policy $\mathbb{R}_\omega$ for fixed offloading decision vector $\boldsymbol{\delta}$ and the optimal collection $\mathbb{R}_{\phi_a}^n, \mathbb{R}_{\phi_j}^n$ of the slices policies, for which constraint \ref{eq:C2} can be satisfied.

\begin{equation}
\min_{\mathbb{R}_\omega} \sum_{n \in \mathcal{N}} \sum_{a \in \mathcal{A}} \frac{1}{\omega_{a}^{n}} \left( \sum_{j \in O_{a,n}(\delta)} \tau_{i,a}^n \right)^{2}
\end{equation}
\begin{align}
\text{} & \quad \textbf{st.} (\ref{eq:C2}),(\ref{eq:C3}), (\ref{eq:C4}), (\ref{eq:C6}) \tag{15a} \label{eq:C14a} 
\end{align}

\textbf{Subproblem 3 (SP3)}:
    As the final step, let us consider the problem of finding the optimal collection $(\mathbb{R}_\omega, \mathbb{R}_{\phi_a}^n, \mathbb{R}_{\phi_j}^n)$ of resource allocation policies first, and finding an optimal offloading decision vector $\boldsymbol{\delta}$ second. 

\begin{equation}\label{eqn_16}
\begin{aligned}
& \underset{\boldsymbol{\delta}}{\text{min}}\underset{{\mathbb{R}_\omega, \mathbb{R}_{\phi_a}^n , \mathbb{R}_{\phi_j}^n}}{\text{min}} \quad C(\boldsymbol{\delta}, \mathbb{R}_\omega, \mathbb{R}_{\phi_a}^n , \mathbb{R}_{\phi_j}^n ) 
\end{aligned}
\end{equation}
\begin{align}
\text{} & \quad \textbf{st.} (\ref{eq:C1}) - (\ref{eq:C7}) \tag{16a} \label{eq:C17a} 
\end{align}



     


\begin{table}[t]
\caption{Symbol descriptions}
    \label{tab:notation}
    \centering
    \begin{tabular}{|p{1.22cm}|p{6.13cm}|}
        \hline
        \textbf{Symbol} & \textbf{Description} \\
        \hline
        $\mathcal{N}$ & Set of application slices \\
        $\mathcal{I}$ & Set of Wireless Devices (WD) \\
        $\mathcal{A}$ & Set of Access Points (AP) \\
        $\mathcal{C}$ & Set of COIN nodes \\
        $\mathcal{M}$ & Set of MEC nodes \\
        $S_i$ & Input size of task $i$ generated by WD $i$ \\
        $R_i$ & Expected number of instructions required for task $i$ \\
        $L_i$ & Number of instructions to execute task $i$ locally \\
        $L_{i,n}$ & Number of instructions to execute task $i$ in slice $n$ \\
        $\delta_i$ & Decision for task $i$, $d_i \in D_i$ \\
        $\Delta_i$ & Set of decisions for task $i$ \\
        $\mathbb{R}_{\omega}$ & Interslice radio resource allocation policy \\
        $\omega_{n}^{a}$ & Interslice radio resource provisioning coefficient \\
        $\mathbb{R}^{n}_{\phi_a}$ & Intraslice radio resource allocation policy \\
        $\phi_{i,a}^{n}$ & Intraslice radio resource provisioning coefficient \\
        $\mathbb{R}^{n}_{\phi_{j}}$ & Intraslice computing resource allocation policy \\
        $\phi_{i,j}^{n}$ & Intraslice computing resource provisioning coefficient \\
        $F_{j}^{n}$ & Computing capability of node $j$ in slice $n$ \\
        $F_{i,j}^{n}(\delta, P_{\phi}^{n})$ & Computing capability allocated to WD $i$ in slice $n$ \\
        $\alpha_i$ & Arrival rate of requests for task $i$ \\
        $F_{i}^{l}$ & Local computing capability of WD $i$ \\
        \hline
    \end{tabular}
\end{table}

In what follows, we propose a hierarchical deepset model in Figure \ref{hierarchical} that solves SP1 for each slice, solves SP2 at the network level, followed by SP3 to find the optimal offloading decision vector.  

The overall architecture follows a similar encoder to ensure consistency across subproblems, while each subproblem is equipped with a dedicated decoder tailored to its objective. The encoder transforms \emph{set} of WD indexed by \(i\in \mathcal{I}\), where the size of \(\mathcal{I}\) varies from set to set. Each WD $i$ contributes a feature vector \(x_i\in\mathbb{O}^{F}\). We combine the WDs features and network status information into a single stacked input matrix
\begin{equation}
X \;=\; 
\begin{bmatrix}
x_1^\top\\
\vdots\\
x_\mathcal{I}^\top
\end{bmatrix}
\in \mathbb{O}^{\mathcal{I}\times F},
\label{eq:stackedX}
\end{equation}
which is a set-valued representation: the row order carries no semantic meaning and may be permuted without changing the underlying set.

Since the WD count of \(I\) differs across sets, we pad each \(X\) to a common set size \(\mathcal{I}_{\max}\) within a mini-batch. Let \(\mathrm{pad}(\cdot)\) append all-zero rows as needed; the padded matrix is
\begin{equation}
\bar{X} \;=\; \mathrm{pad}(X) \in \mathbb{O}^{\mathcal{I}_{\max}\times F}.
\label{eq:paddedX}
\end{equation}
We define \(\bar{X}\) with a binary mask that indicates which rows correspond to real WD and which are padding:
\begin{equation}
m \;=\; (m_1,\dots,m_{\mathcal{I}_{\max}})^\top \in \{0,1\}^{\mathcal{I}_{\max}}
\end{equation}
\begin{equation}
m_i \;=\; 
\begin{cases}
1, & \text{$i$ = real WD$i$} \\
0, & \text{$i$ = padding}
\end{cases}
\label{eq:mask}
\end{equation}
This mask is used throughout the model to: (i) exclude padded rows from set summaries and attention, and (ii) prevent padded rows from receiving probability mass.

We standardize features to stabilize optimization. Let \(\mu\in\mathbb{O}^{F}\) and \(\sigma\in\mathbb{O}^{F}\) denote the per-feature mean and standard deviation computed over the \emph{training} set. We transform each padded set by
\begin{equation}
\tilde{X}_{i,:} \;=\; (\bar{X}_{i,:} - \mu) \oslash \sigma \;\in\; \mathbb{O}^{F},
\qquad i=1,\dots,\mathcal{I}_{\max}
\label{eq:standardize}
\end{equation}
where \(\oslash\) denotes element-wise division; padded set remain zeros after standardization.

Given \((\tilde{X},m)\) for a set with \(\mathcal{I}\) real WD $i$ and a fixed set of allocation policy \(\mathbb{R}\) indexed by \(e\in\{1,\dots,\mathcal{E}\}\), our objective is to learn a permutation-equivariant mapping that outputs, for each \(\mathbb{R}\), a probability distribution over \(I\) \emph{plus} a slack (unused capacity) entry, summing to one per \(\mathbb{R}\). Formally,
\begin{equation}
f_{\theta}:\;\mathbb{O}^{\mathcal{I}_{\max}\times F}\times\{0,1\}^{\mathcal{I}_{\max}}
\;\longrightarrow\;
[0,1]^{(\mathcal{I}+1)\times E}
\label{eq:goal}
\end{equation}

\begin{equation}
\sum_{u=0}^{\mathcal{I}} P_{u,e} \;=\; 1 \;\; \forall e
\label{eq:goal}
\end{equation}

where \(u=1,\dots,\mathcal{I}\) index real users and \(u=0\) denotes slack. The mapping must respect user-set permutation symmetry and operate for arbitrary \(\mathcal{I}\le \mathcal{I}_{\max}\).

\subsection{DeepSet Encoder}

Given the standardized and padded user matrix \(\tilde{X}\in\mathbb{O}^{\mathcal{I}_{\max}\times F}\) and mask \(m\in\{0,1\}^{\mathcal{I}_{\max}}\) from \eqref{eq:paddedX}--\eqref{eq:standardize}, the encoder produces a \emph{contextual} embedding for each WD row while remaining permutation–equivariant in the WD dimension. The encoder follows a DeepSets-style architecture with a row-wise embedding network \(\phi\) and a context-fusing network \(\rho\).

A shared multilayer perceptron (MLP) \(\phi:\mathbb{O}^{F}\!\to\!\mathbb{R}^{H}\) maps each user’s features to a hidden vector:
\begin{equation}
h_i \;=\; \phi\!\big(\tilde{X}_{i,:}\big) \in \mathbb{O}^{H}, 
\qquad i=1,\dots,\mathcal{I}_{\max}
\label{eq:phi}
\end{equation}
Stacking the row-wise outputs yields \(H_{\text{row}}=\big[h_1^\top;\dots;h_{\mathcal{I}_{\max}}^\top\big]\in\mathbb{O}^{\mathcal{I}_{\max}\times H}\)

To capture global context from the set of (real) WD, we compute a masked mean over the hidden vectors, excluding padded rows:
\begin{equation}
c \;=\; 
\frac{\sum_{i=1}^{\mathcal{I}_{\max}} m_i\, h_i}{\sum_{i=1}^{\mathcal{I}_{\max}} m_i \;+\; \varepsilon}
\;\in\; \mathbb{O}^{H},
\label{eq:context}
\end{equation}
where \(\varepsilon>0\) avoids division by zero in degenerate cases. The same context vector \(c\) is broadcast to every row.

\paragraph*{Fuse local and global information.}
Each user combines its local representation \(h_i\) with the global context \(c\) via concatenation followed by a second shared MLP \(\rho:\mathbb{O}^{2H}\!\to\!\mathbb{O}^{D}\):
\begin{equation}
z_i \;=\; \rho\!\big([\,h_i \,;\, c\,]\big) \in \mathbb{O}^{D},
\qquad i=1,\dots,\mathcal{I}_{\max}.
\label{eq:rho}
\end{equation}
Stacking the resulting vectors forms the encoder output matrix
\begin{equation}
Z \;=\;
\begin{bmatrix}
z_1^\top\\
\vdots\\
z_{\mathcal{I}_{\max}}^\top
\end{bmatrix}
\in \mathbb{O}^{\mathcal{I}_{\max}\times D}.
\label{eq:Z}
\end{equation}
Rows with \(m_i=0\) correspond to padded users; these rows are carried forward but ignored by all masked operations downstream.

Let \(\Pi\in\{0,1\}^{\mathcal{I}_{\max}\times \mathcal{I}_{\max}}\) be a permutation matrix acting on WD rows, and define \(\Pi m\) as the permuted mask. Because \(\phi\) and \(\rho\) act row-wise and the summary in \eqref{eq:context} is a symmetric (mean) operator, the encoder is permutation–equivariant:
\begin{equation}
\mathrm{Enc}(\Pi \tilde{X}, \Pi m) \;=\; \Pi\, \mathrm{Enc}(\tilde{X}, m).
\label{eq:equiv}
\end{equation}

\subsection{Decoder}

For SP1, given the encoder output \(Z\in\mathbb{O}^{\mathcal{I}_{\max}\times D}\) in \eqref{eq:Z} and the mask \(m\in\{0,1\}^{\mathcal{I}_{\max}}\), the decoder maps each WD $i$ embedding to per–resource scores and converts them into valid WD–wise probability distributions for every \(\mathbb{O}\), augmented with a slack term to capture unused capacity. Let the resource index set be \(\mathcal{E}=\{1,\dots,E\}\).

For each \(e\in\mathcal{E}\), the decoder maintains learned parameters
\begin{equation}
w_e\in\mathbb{O}^{D},\qquad b_e\in\mathbb{O},\qquad s_e\in\mathbb{O},
\label{eq:dec-params}
\end{equation}
shared across all users. The WD logits and the slack logit are
\begin{equation}
l_{i,e} \;=\; w_e^\top z_i + b_e,\quad i=1,\dots,\mathcal{I}_{\max},
\quad
l_{0,e} \;=\; s_e
\label{eq:logits}
\end{equation}
Stacking over WDs yields \(A\in\mathbb{O}^{\mathcal{I}_{\max}\times E}\) with entries \([A]_{i,e}=a_{i,e}\).

To ignore padded rows during normalization, we define augmented masked logits \(\tilde{L}\in\mathbb{O}^{(\mathcal{I}_{\max}+1)\times E}\) by
\begin{equation}
\tilde{l}_{i,e} \;=\;
\begin{cases}
l_{i,e}, & m_i=1,\\
-\infty, & m_i=0,
\end{cases}
\quad i=1,\dots,\mathcal{I}_{\max},
\quad
\tilde{l}_{0,e}=l_{0,e}
\label{eq:masked-logits}
\end{equation}

For each \(e\in\mathcal{E}\), we normalize along the augmented WD axis:
\begin{equation}
\begin{split}
p_{i,e}
\;=\;
\frac{\exp\!\big(\tilde{l}_{i,e}\big)}
{\displaystyle \sum_{j\in\{0,1,\dots,\mathcal{I}_{\max}\}} \exp\!\big(\tilde{l}_{j,e}\big)},
\label{eq:user-softmax}
\end{split}
\end{equation}
Collecting all \(p_{i,e}\) gives
\begin{equation}
P \;\in\; [0,1]^{(\mathcal{I}_{\max}+1)\times E},
\qquad
\sum_{i=0}^{\mathcal{I}_{\max}} p_{i,e} \;=\; 1,\;\; \forall e\in\mathcal{E}.
\label{eq:P-stochastic}
\end{equation}
The rows \(i=1,\dots,\mathcal{I}_{\max}\) correspond to actual WDs; the row \(i=0\) corresponds to slack.

By \eqref{eq:P-stochastic}, each resource satisfies
\begin{equation}
\sum_{i=1}^{\mathcal{I}_{\max}} \hat{y}_{i,e} \;=\; \,\bigl(1-p_{0,e}\bigr)\;\le\; 1,
\qquad \forall e\in\mathcal{E},
\label{eq:sum-constraint}
\end{equation}

 thus enforcing constraint \eqref{eq:C5} as proved in Lemma~\ref{lem:13e}. Moreover, since \eqref{eq:logits} is row-wise and the normalization \eqref{eq:user-softmax} acts symmetrically across WD rows (modulo masking), the decoder is permutation–equivariant in the WD dimension; combined with \eqref{eq:equiv}, the overall encoder–decoder mapping preserves user permutations.

Let $Y\in[0,1]^{(\mathcal{I}_{\max}+1)\times E}$ denote the soft labels obtained by scaling the outputs to $[0,1]$. For each $e\in\mathcal E$, set
\begin{equation}
y_{0,e}\;=\;\max\!\Big\{0,\;1-\sum_{i=1}^{\mathcal{I}} y_{i,e}\Big\},
\qquad
\sum_{i=0}^{\mathcal{I}_{\max}} y_{i,e}\;=\;1,
\label{34}
\end{equation}
thus, the SP1 training objective is the masked cross-entropy on the augmented simplex:
\begin{equation}
\mathcal{L}_{\mathrm{SP1}}
\;=\;
-\sum_{e\in\mathcal E}\;\sum_{i=0}^{\mathcal{I}_{\max}}
\bar m_i\, y_{i,e}\,\log p_{i,e}.
\label{35}
\end{equation}

\begin{lemma}[]\label{lem:13e}
For every slice $n$ and resource $e$, the allocations of resources per-WD $\{\phi^{n}_{i,e}\}_{i \in \mathcal{I}}$ produced by SP1 satisfy
\[
\sum_{i \in \mathcal{I}} \phi^{n}_{i,e} \le 1.
\]
\end{lemma}

\begin{proof}
The SP1 decoder generates, for each resource $e$ in slice $n$, an augmented softmax distribution over the index set $\{0\} \cup \mathcal{I}$, where index $0$ denotes a slack option corresponding to unused capacity. 
Let the resulting probabilities be $p_0, p_1, \ldots, p_{|\mathcal{I}|}$ with $p_k \ge 0$ and $\sum_k p_k = 1$. 
Define $\phi^{n}_{i,e} = p_i$ for each user $i$. Then,
\[
\sum_{i \in \mathcal{I}} \phi^{n}_{i,e}
= \sum_{i \in \mathcal{I}} p_i
= 1 - p_0
\le 1.
\]
Hence, the total resource share allocated to users within each slice never exceeds the available capacity of resource $e$, satisfying constraint (13e). \qedhere
\end{proof}

For SP2, the outputs are set-level and identical for all WDs in that set. We therefore first aggregate the encoder states by a masked mean to obtain a single set embedding
\begin{equation}
g \;=\; \frac{\sum_{i=1}^{\mathcal{I}_{\max}} m_i\, z_i}{\sum_{i=1}^{\mathcal{I}_{\max}} m_i + \varepsilon} \;\in\; \mathbb{O}^{D}
\label{36}
\end{equation}
The SP2 decoder maintains learned parameters
$w_e \in \mathbb{O}^{D},\quad b_e \in \mathbb{O} \ \ \text{for } e\in\mathcal E,\quad s \in \mathbb{O} \ $, and forms logits from the pooled representation:
\begin{equation}
a_e \;=\; w_e^\top g + b_e,\quad e\in\mathcal E,\qquad a_0 \;=\; s \
\end{equation}
A softmax over the augmented slice axis $\{0\}\cup\mathcal E$ yields a normalized distribution
\begin{equation}
p_k \;=\; \frac{\exp(a_k)}{\sum_{\ell\in\{0\}\cup\mathcal E}\exp(a_\ell)},\quad k\in\{0\}\cup\mathcal E
\end{equation}
with predicted slice allocations $\hat y_e = p_e$ for $e\in\mathcal E$ and unused capacity $p_0$. By construction,
\begin{equation}
\sum_{e\in\mathcal E} \hat y_e \;=\; 1 - p_0 \;\le\; 1
\label{39}
\end{equation}
as proved in Lemma~\ref{lem:13d}, \eqref{39} enforces constraint \eqref{eq:C4}. Training uses soft labels $y=(y_e)_{e\in\mathcal E}\in[0,1]^E$, augmented with a slack target $y_0=\max\!\big\{0,\,1-\sum_{e\in\mathcal E} y_e\big\}$. The loss is the cross-entropy on the augmented simplex:
\begin{equation}
\mathcal L_{\mathrm{SP2}} \;=\; - \sum_{k\in\{0\}\cup\mathcal E} y_k \log p_k.
\end{equation}
Since \eqref{36}--\eqref{39} depend only on the pooled set embedding $g$, the SP2 outputs are identical for every WD in the set, as required, while the shared encoder and masked operations preserve permutation equivariance in the WD dimension.

\begin{lemma}[]\label{lem:13d}
For each Access Point $a$, the slice shares $\{\omega_a^{n}\}_{n \in \mathcal{N}}$ produced by SP2 satisfy
\[
\sum_{n \in \mathcal{N}} \omega_a^{n} \le 1.
\]
\end{lemma}

\begin{proof}
SP2’s decoder outputs an augmented softmax over the set $\{0\}\cup\mathcal{N}$, where index $0$ is a slack option. Let the softmax probabilities be $p_0, p_1, \ldots, p_{|\mathcal{N}|}$ with $p_k \ge 0$ and $\sum_k p_k = 1$. Define the usable slice shares as $\omega_a^{n} := p_n$ for each slice $n \in \mathcal{N}$. Then
\[
\sum_{n \in \mathcal{N}} \omega_a^{n}
= \sum_{n \in \mathcal{N}} p_n
= 1 - p_0
\le 1.
\]
Hence, the total per-AP allocation across slices never exceeds the AP’s full capacity, and (13d) holds. \qedhere
\end{proof}

Finally, for SP3, each WD \(i\) is assigned an offloading decision \(\delta_i \in \Delta_i\), which is either local execution or offloading through AP \(a\) to node \(j\) in slice \(n\).
Given the encoder output in (26), we use one joint head that scores the exact decision in \(\Delta_i\).
All heads share parameters across WDs and act row-wise.
With \(z_i \in \mathbb{O}^{D}\), the WD logits and probabilities are
\begin{equation}
\ell_i \;=\; W^{(\delta)} z_i + b^{(\delta)} \;\in\; \mathbb{O}^{|\Delta_i|}, 
\end{equation}
\begin{equation}
p_i \;=\; \mathrm{softmax}(\ell_i),\;\; \mathbf{1}^\top p_i = 1,
\end{equation}
and the predicted decision is the joint argmax
\begin{equation}
\hat{\delta}_i \;=\; \arg\max_{\delta \in \Delta_i}\,[p_i]_{\delta},
\end{equation}
which maps to a specific decision \((n,a,j)\) when \(\hat{\delta}_i\neq\text{local}\) as proved in Lemma \ref{lem:13a}.

To improve sample efficiency, we also include:
(i) a binary local/offload head
\begin{equation}
q_i \;=\; \sigma\!\big(w_{\mathrm{bin}}^\top z_i + b_{\mathrm{bin}}\big) \;\in\; (0,1),
\end{equation}
and (ii) separate $n/a/j$ heads
\begin{subequations}
\begin{align}
n_i \;&=\; W^{(n)} z_i + b^{(n)} \;\in\; \mathbb{O}^{|\overline{N}|}, 
& p_i^{(n)} &= \mathrm{softmax}(n_i) \\[-1pt]
a_i \;&=\; W^{(a)} z_i + b^{(a)} \;\in\; \mathbb{O}^{|A|}, 
& p_i^{(a)} &= \mathrm{softmax}(a_i) \\[-1pt]
j_i \;&=\; W^{(j)} z_i + b^{(j)} \;\in\; \mathbb{O}^{|J|}, 
& p_i^{(j)} &= \mathrm{softmax}(j_i)
\end{align}
\end{subequations}

We use the joint argmax over \(\Delta_i\); an equivalent rule is

\begin{equation}
\hat{\delta}_i =
\begin{cases}
\text{local}, & q_i < \tfrac{1}{2},\\[4pt]
(\arg\max p_i^{(n)},\; \arg\max p_i^{(a)},\\
\qquad \arg\max p_i^{(j)}), & \text{otherwise.}
\end{cases}
\end{equation}

Let \(y_i^{\delta}\in\{0,1\}^{|\Delta_i|}\) be the joint label, \(y_i^{\mathrm{bin}}\in\{0,1\}\) the local/offload flag, and \(y_i^{(n)}\), \(y_i^{(a)}\), \(y_i^{(j)}\) the $n/a/j$ labels for offloaded WDs.
With the padding mask \(m_i\) from \eqref{eq:mask}, we minimize
\begin{equation}
\mathcal{L}_{\text{joint}}
=\;-\sum_i m_i\,\big\langle y_i^{\delta},\,\log p_i\big\rangle
\end{equation}
\begin{equation}
\mathcal{L}_{\text{bin}}
=\;-\sum_i m_i\!\left(y_i^{\mathrm{bin}}\log q_i + (1-y_i^{\mathrm{bin}})\log(1-q_i)\right)
\end{equation}
\begin{equation}
\begin{split}
\mathcal{L}_{\text{n,a,j}}
=\;-\sum_i m_i\,y_i^{\mathrm{bin}}\Big(\big\langle y_i^{(n)},\,\log p_i^{(n)}\big\rangle \\
+\big\langle y_i^{(a)},\,\log p_i^{(a)}\big\rangle
+\big\langle y_i^{(j)},\,\log p_i^{(j)}\big\rangle\Big)
\end{split}
\end{equation}
and combine them (weights \(\lambda_{\mathrm{bin}},\lambda_{\mathrm{fac}}\ge 0\); set to 1 unless stated otherwise):
\begin{equation}
\mathcal{L}_{\text{SP3}}
\;=\;
\mathcal{L}_{\text{joint}} + \lambda_{\mathrm{bin}}\,\mathcal{L}_{\text{bin}} + \lambda_{\mathrm{fac}}\,\mathcal{L}_{\text{n,a,j}}.
\end{equation}
Because all mappings above act row-wise with shared parameters and all normalizations respect the mask, the decoder is permutation–equivariant in the WD dimension.

\begin{lemma}[]\label{lem:13a}
Each WD $i$ selects exactly one decision variable $\delta_i$ from $\Delta_i$ .
\end{lemma}

\begin{proof}
Let the SP3 classifier output a probability vector 
$p_i = [p_i^{1}, p_i^{2}, \ldots, p_i^{K}] \in \mathbb{R}^{K}$, 
where $K = |\Delta_i|$ is the number of admissible offloading decisions for WD $i$. 
Since the final decision is taken as 
$\widehat{\delta}_i = \arg\max_{k} p_i^{k}$, 
exactly one index is active while all others are inactive. 
Therefore,
\[
\sum_{k=1}^{K} \widehat{\delta}_i^{k} = 1,
\]
satisfying constraint (13a). \qedhere
\end{proof}

\section{Experiments and Evaluation}

\subsection{Dataset Generation}

The datasets used to train and evaluate the proposed DeepSets-S model were generated by solving the optimization problem formulated in Section~IV using standard nonlinear solvers. Each solution captures the interactions among wireless device (WD) requests, slice configurations, access points (APs), and edge nodes (COIN and MEC). A total of 5,000 simulation runs were conducted under varying user sizes, bandwidth distributions, and slice configurations, resulting in three distinct datasets corresponding to the decomposed sub-problems introduced in Section~V:

\begin{itemize}
    \item \textbf{SP1 dataset:} Takes the user request information as input and produces the optimal intra-slice radio and computing allocation coefficients for each slice as output.
    \item \textbf{SP2 dataset:} Uses the inputs and outputs from SP1 as its input and yields the optimal inter-slice bandwidth allocation coefficients across APs as output.
    \item \textbf{SP3 dataset:} Uses the inputs and outputs from SP2 as its input and produces the optimal offloading decisions as output.
\end{itemize}

For evaluation, we adopted a 10-fold cross-validation procedure to ensure robustness and prevent overfitting.

\subsection{Experiment Environment}

We consider a small-scale scenario for an edge application with a square area of (500 x 500)$m$ where 8 COINs, 1 MEC, 3 APs, and WDs are placed at random. The power spectral density of Gaussian noise $N_0$ is $-174\ dBm/Hz$ \cite{aliyu2024digital}, and according to \cite{jovsilo2022joint}, we represent the achievable physical rate $R_{i,a}$ from WD $i$ to AP $a$ as \( R_{i,a} =B_alog\  (1 + (distance) P^t/N_0) \), where $distance$ is the Euclidean distance between WD $i$ and AP $a$, and $P^t$ is the transmission power of WD $i$. Other simulation parameters are summarized in table \ref{table:sim_params}.

\begin{table}[!t]
\renewcommand{\arraystretch}{1.2}
\caption{Simulation Parameters}
\label{table:sim_params}
\centering
\begin{tabular}{l|l}
\hline
\textbf{Parameter} & \textbf{Value} \\
\hline\hline
COINs node capacity  & $[5, 10]$ x $10^8$ $CPU\ cycles/s$ \cite{lia2022}\\
MEC capacity & $1$ x $10^{10}$ $CPU\ cycles/s$ \cite{lia2022}\\
WD local computing capacity & $2$ x $10^{7}$ $CPU\ cycles/s$ \cite{sardellitti2015joint}\\
Input task size & $[1, 10]$ MB \cite{aliyu2024digital}\\
Bandwidth $B_a$ per AP & 18 MHz \cite{ahmadi2009overview}\\
Transmit Power of WD's & $[10^{-6} , 0.1]W$ \cite{jovsilo2022joint}\\
$distance$ & 4 \\
\hline
\end{tabular}
\end{table}



We conducted a comprehensive evaluation of our DeepSets-S model against other ablations. The models were trained on three set of datasets. The data were generated by solving the previously defined MINLP subproblems with standard optimizers. We performed 300 iterations for each sample request with a $90\%$ confidence interval. Each dataset corresponds to one of the subproblems introduced in Section \ref{ML-model}: SP1 dataset, SP2 dataset, and SP3 dataset. 

In Table \ref{tab:sp1_sp2_perf}, we evaluate the performance of our DeepSets-S model against its variant DeepSets which disables slack, and apply normalization on the data to satisfy the resource allocation constraint, and other state-of-the-art multi-instance learning model; SetTransformer \cite{lee2019set}. We evaluate using Acc@k, where Acc@0.5, Acc@1, and Acc@2 measure the proportion of predictions that fall within 0.5, 1, and 2 units of the ground truth, respectively. 

On SP1, DeepSets-S attains an Acc@0.5 of 76.76\%, which is higher than both the no-slack DeepSets baseline (76.74\%) and the SetTransformer variants (76.71–76.75\%). At the threshold of Acc@1, DeepSets-S records 95.26\%, outperforming DeepSets (95.21\%) and SetTransformer (95.22–95.23\%). For the more relaxed Acc@2, DeepSets-S reaches 99.12\%, surpassing DeepSets (99.04\%) and the SetTransformer models (99.04–99.05\%). However, the SetTransformer with 4 heads achieves the lowest MAE (0.3918) and the highest R² (0.9059).

For SP2 also, DeepSets-S attains the best tolerance-based accuracy across all thresholds. It achieves an Acc@0.5 of 81.85\%, slightly higher than the baseline without slack DeepSets (81.83\%) and the SetTransformer models (81.80–81.83\%). At Acc@1, DeepSets-S records 95. 67\%, exceeding DeepSets (95. 64\%) and close to the SetTransformer results (95.66\%). For the more relaxed Acc@2, DeepSets-S achieves 99.08\%, matching the best SetTransformer variant and surpassing DeepSets (99.05\%). In regression measures, the four-head SetTransformer reports the lowest MAE (0.3474) and the highest R² (0.9674), whereas DeepSets-S achieves 0.3480 and 0.9668.

Figures \ref{fig:AE_a} and \ref{fig:AE_b} present the absolute error of the DeepSets-S model with respect to the number of WDs for SP1 and SP2, respectively. For both SP1 and SP2, the error is higher when the WD size is small, reflecting the limited availability of input information. As the WD size increases, the error decreases steadily and stabilizes at values close to zero, indicating that the model becomes more reliable with larger WD groups. This trend confirms that DeepSets-S generalizes well across varying set sizes and maintains consistent accuracy in large-scale scenarios.

For the SP3 dataset, we evaluate multiclass accuracy (Acc) for the offloading decision vector $\delta$ and binary accuracy (Acc\_bin) to distinguish between local and offload decisions. The proposed DeepSets-S model achieves an Acc of 0.7486 and an Acc\_bin of 0.8824, showing improvements over the no-slack DeepSets, which records an Acc of 0.7227 and an Acc\_bin of 0.8802. The two-head SetTransformer obtains an Acc of 0.7454 and an Acc\_bin of 0.8581, while the four-head variant reaches the highest Acc of 0.7492 but a lower Acc\_bin of 0.8733 compared with DeepSets-S. These results highlight that the slack mechanism in DeepSets-S provides a more balanced trade-off.

\begin{table}[ht]
\centering
\caption{Performance of DeepSets-S Models on SP1 and SP2 against other ablations (best results are bold).}
\label{tab:sp1_sp2_perf}
\scriptsize
\begin{tabular*}{\linewidth}{@{\extracolsep{\fill}}lccccc}
\toprule
\textbf{Models} & \textbf{MAE (\%)} & \textbf{Acc@0.5} & \textbf{Acc@1} & \textbf{Acc@2} & \textbf{$R^2$} \\
\midrule
\multicolumn{6}{c}{\textbf{SP1 Dataset}} \\
\midrule
\textbf{DeepSets-S}             & \textbf{0.3924} & \textbf{0.7676} & \textbf{0.9526} & \textbf{0.9912} & 0.9057 \\
DeepSets                       & 0.3920 & 0.7674 & 0.9521 & 0.9904 & 0.9057 \\
SetTransformer (2 heads)       & 0.3922 & 0.7671 & 0.9522 & 0.9904 & 0.9055 \\
SetTransformer (4 heads)       & 0.3918 & 0.7675 & 0.9523 & 0.9905 & \textbf{0.9059} \\
\\[-0.5ex]
\midrule
\multicolumn{6}{c}{\textbf{SP2 Dataset}} \\
\midrule
\textbf{DeepSets-S}             & \textbf{0.3480} & \textbf{0.8185} & \textbf{0.9567} & \textbf{0.9908} & 0.9668 \\
DeepSets                       & \textbf{0.3480} & 0.8183 & 0.9564 & 0.9905 & 0.9668 \\
SetTransformer (2 heads)       & 0.3475 & 0.8180 & 0.9566 & 0.9905 & 0.9673 \\
SetTransformer (4 heads)       & 0.3474 & 0.8183 & 0.9566 & \textbf{0.9908} & \textbf{0.9674} \\
\bottomrule
\end{tabular*}
\end{table}

\begin{table}[ht]
\centering
\caption{Performance of DeepSets-S model against other ablations on SP3 (best results are bold). Acc (Multiclass accuracy for offloading decision vector $\boldsymbol{\delta}$), Acc\_bin (Binary accuracy for local and offload decisions)}
\label{tab:sp3_perf}
\scriptsize
\begin{tabular*}{\linewidth}{@{\extracolsep{\fill}}lcc}
\toprule
\textbf{Models} & \textbf{Acc} & \textbf{Acc\_bin} \\
\midrule
\textbf{DeepSets-S}               & 0.7486 & \textbf{0.8824} \\
DeepSets                       & 0.7227 & 0.8802 \\
SetTransformer (2 heads)       & 0.7454 & 0.8581 \\
SetTransformer (4 heads)       & \textbf{0.7492} & 0.8733 \\
\bottomrule
\end{tabular*}
\end{table}

\begin{figure}[!t]
  \centering
  \begin{subfigure}{0.8\linewidth}
    \centering
    \includegraphics[width=\linewidth]{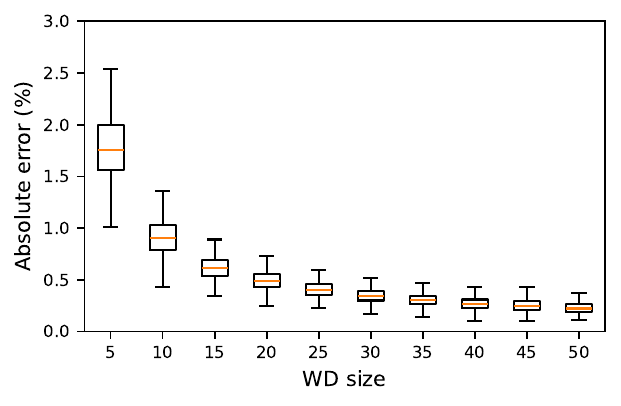}
    \caption{SP1 dataset}
    \label{fig:AE_a}
  \end{subfigure}

  \vspace{0.6em} 

  \begin{subfigure}{0.8\linewidth}
    \centering
    \includegraphics[width=\linewidth]{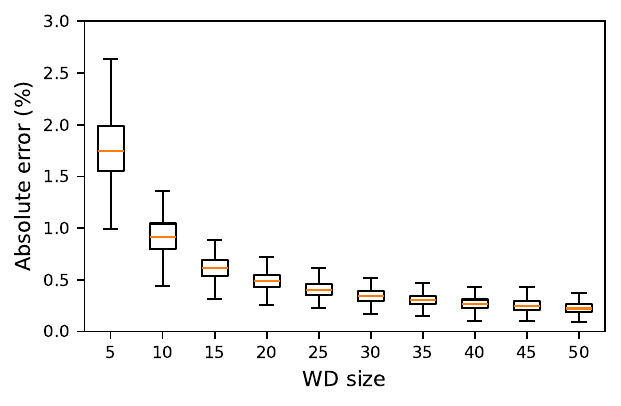}
    \caption{SP2 dataset}
    \label{fig:AE_b}
  \end{subfigure}

  \caption{Absolute error of DeepSet-S model against WD size.}
  \label{fig:Absolute_Error}
\end{figure}

In what follows, we present the results of the proposed DeepSets-S solution compared to some benchmark resource allocation models. Specifically, we employ the following models for evaluation:

\begin{itemize}
    \item \textbf{Optimization Solvers (OS):} We solve \textit{SP1} and \textit{SP2} through Ipopt optimization solver \cite{wachter2006implementation}, and \textit{SP3} through Gurobi optimization solver \cite{gurobi} to achieve the optimal offloading decision vector $\boldsymbol{\delta}$.
    \item \textit{Ours} \textbf{DeepSets-S:} Employs DeepSets-S model that trains on the solution achieved through OS to solve SP1, SP2 and SP3.
    \item \textbf{Proportional Resource Allocation Policy (PRAP):} We define the policies $\mathbb{R}_{\omega}^{pr}$, which share the bandwidth of each AP $a$ to each slice $n$ proportionally, $\mathbb{R}_{\phi_a}^{pr}$ which gives propotional share of the bandwidth of each AP $a$ in slice $n$ to each WD $i$, and finally $\mathbb{R}_{\phi_j}^{pr}$ which gives proportional share of the computing resource of node $j$ in slice $n$ to each WD $i$.
    \item \textbf{Equal Resource Allocation Policy (ERAP):} We define the policies $\mathbb{R}_{\omega}^{eq}$, which gives equal share of the bandwidth of each AP $a$ to each slice $n$, $\mathbb{R}_{\phi_a}^{eq}$ which gives equal share of the bandwidth of each AP $a$ in slice $n$ to each WD $i$, and finally $\mathbb{R}_{\phi_j}^{eq}$ which gives equal share of the computing resource of node $j$ in slice $n$ to each WD $i$.
\end{itemize}

To evaluate the practical computational complexity of the solution, we conducted an empirical analysis by executing each model type with varying user request rates. The experiments were performed using the Python programming language, and the computation time was measured as the elapsed time between the start and end of execution.

Figure \ref{Algorithm Computation Time} illustrates the computation time (on a logarithmic scale) as a function of the request rate. As shown, the OS exhibits the highest computational burden, with computation time increasing with the number of requests. This is due to the exhaustive search to achieve the optimal solution by the standard solvers. In contrast, the DeepSets-S model, once trained, achieves significantly lower execution time, resulting in a 86.1\% reduction in computational time compared to the OS. Evethough, the ERAP and PRAP make quicker resource allocations, they still have to rely on the solver for maing offloading decisions, and hence more computation time than the DeepSets-S model.

The simulations were conducted on a system equipped with a 12th Generation Intel(R) Core(TM) i5-12400F processor running at 2.50 GHz, 16 GB of RAM, and a 500 GB hard drive.

\begin{figure}[!ht]
  \centering
  \includegraphics[width= 0.9\linewidth]{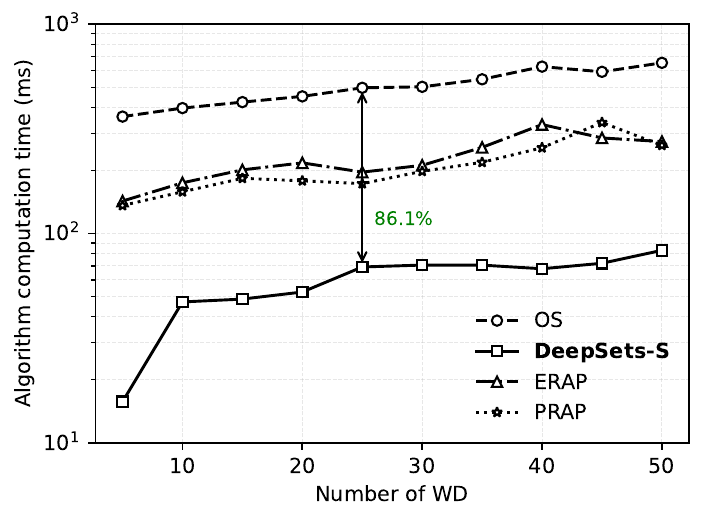}
  \caption{Algorithm computation time}
  \label{Algorithm Computation Time}
\end{figure}

\begin{figure}[!ht]
  \centering
  \includegraphics[width= 1.0\linewidth]{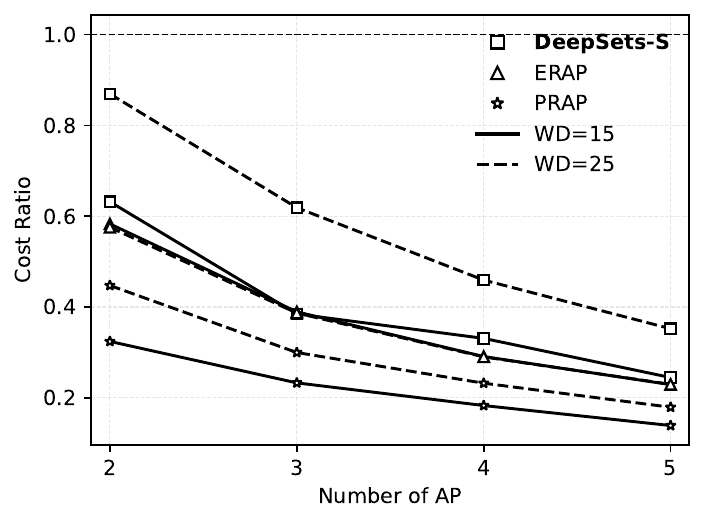}
  \caption{Cost Ratio vs Number of APs}
  \label{CR_AP}
\end{figure}

To evaluate the overall efficiency of our proposed framework, we define the system cost ratio $CR$ as a key performance metric. The $CR$ quantifies the ratio between the system cost reached using the optimization solver ($C_{OS}$) and the system cost reached under the baseline model ($C_{m}$).  

\begin{equation}
    \text{CR} = \frac{C_{\text{OS}}}{C_m}, \quad m \in \{\text{DeepSets-S}, \text{ERAP}, \text{PRAP}\}
\end{equation}

A higher $CR$ indicates that the model is more effective at minimizing system costs and is performing closer to the optimal solution.

Figure \ref{CR_AP} illustrates the $CR$ as a function of the number of APs for WD sizes of 15 and 25. In both scenarios, the proposed DeepSets-S model consistently achieves higher cost ratios than the baselines ERAP and PRAP, confirming its ability to approximate the optimal solution more closely. The performance gap is more pronounced for larger workloads (WD = 25), where DeepSets-S maintains a substantial advantage.

\begin{figure}[!ht]
  \centering
  \includegraphics[width= 1.0\linewidth]{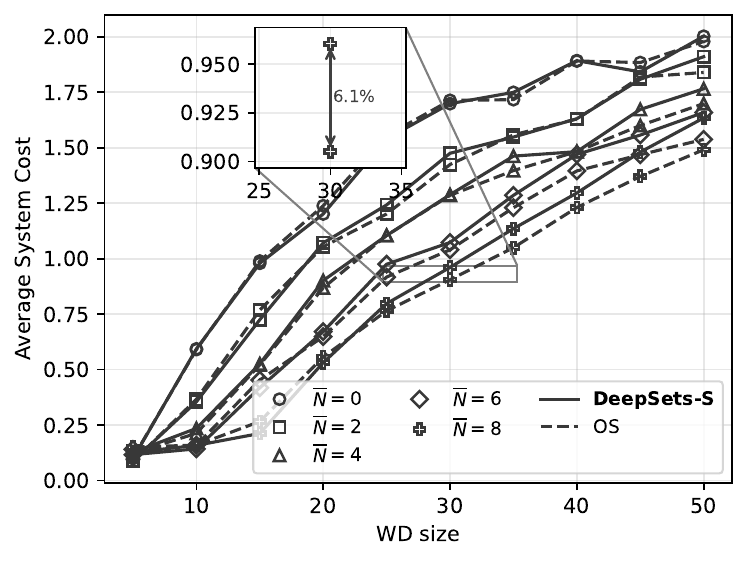}
  \caption{System Cost vs WD size for Different COIN Size}
  \label{system cost for coin}
\end{figure}

Next, Figure \ref{system cost for coin} presents the average system cost by WD size for different numbers of COINS. As expected, the system cost increases with the number of WDs due to the higher aggregate workload in the network. Across all COIN configurations, the proposed DeepSets-S model closely follows the performance of the OS, demonstrating its ability to generalize to larger set of data. The inset plot highlights that the deviation between DeepSets-S and the optimizer remains small, within 6.1\% even at WD size of 30, confirming the effectiveness of the model in approximating near-optimal allocations. Moreover, increasing the number of COINs consistently lowers the system cost, as COINs reduce the workload on the MEC in queueing and transmission delays, and enables more favorable offloading placements, thereby improving overall efficiency.

Figure \ref{fig:system_cost_slice} illustrates the variation of average system cost with respect to the number of WDs under different slice configurations for both the DeepSets-S model and the OS. As the number of WDs increases, the system cost correspondingly rises due to higher aggregated workloads and contention across communication and computing resources. For DeepSets-S, the increase follows a similar trend as the OS, showing that the learned model maintains stable generalization as the task load scales. In all slice configurations, the cost values of DeepSets-S remain very close to those of the OS, with an average difference of just 3.8\% at 30-50 WDs, confirming the model’s ability to approximate near-optimal allocations effectively.

\begin{figure}[!ht]
  \noindent
  \hspace*{-1.5em}
  \includegraphics[width=1.1\linewidth]{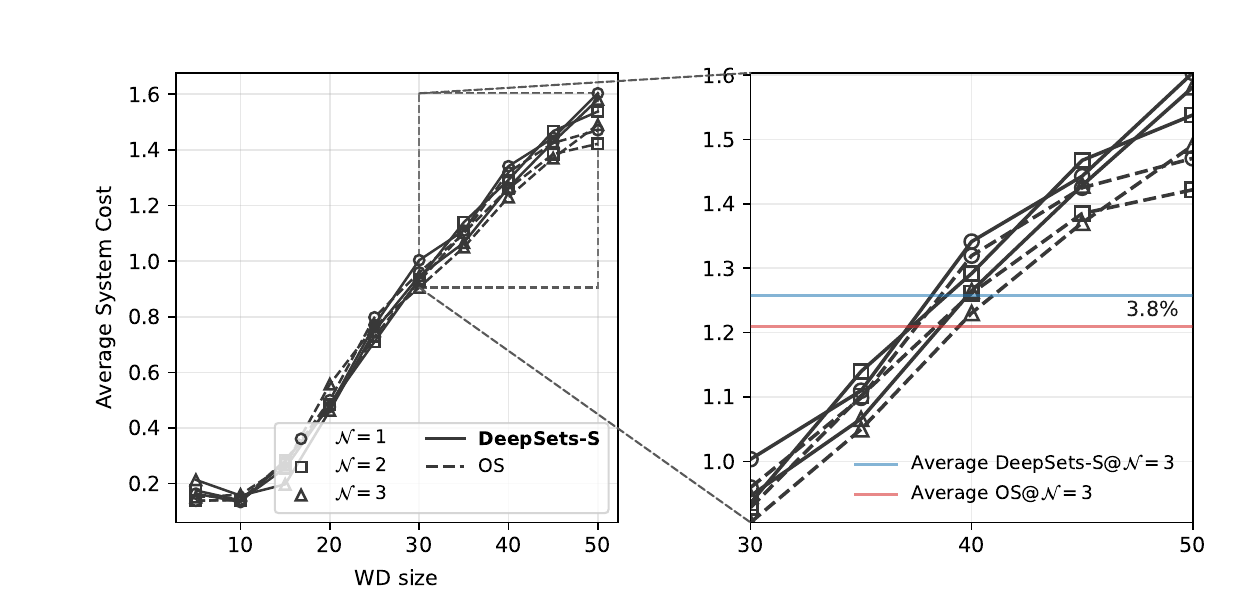}
  \caption{System Cost vs WD size for Different Slice Size}
  \label{fig:system_cost_slice}
\end{figure}

Furthermore, the slight cost reduction observed when the number of slices increases reflects the improved adaptability of the slicing framework. With more slices, tasks can be more evenly distributed across available COIN and MEC nodes, reducing queuing delays and improving parallel resource utilization.

Figure \ref{percentage offload} presents the percentage of tasks offloaded to the MEC out of all offloaded tasks as the number of WDs increases, under different slice configurations. For both the OS and the proposed DeepSets-S model, the percentage of MEC-offloaded tasks gradually increases with the number of WDs. This rise indicates that as system load grows, a larger fraction of the offloaded workload is directed to the MEC.

DeepSets-S closely mirrors the OS trend in every slice configuration, with only slight deviations at higher WD sizes, confirming that the learned model replicates the optimizer’s offloading pattern with high fidelity.

\begin{figure}[!ht]
  \centering
  \includegraphics[width= 1.0\linewidth]{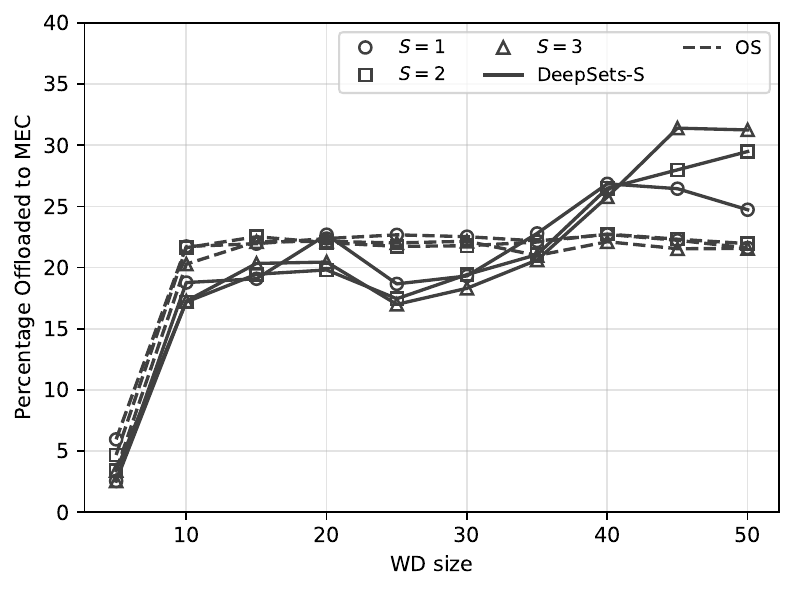}
  \caption{Percentage of tasks offloaded to MEC}
  \label{percentage offload}
\end{figure}


\section{Discussions and Limitations}

So far, we have shown that the formulated MINLP problem can be decomposed into three coupled subproblems—SP1, SP2, and SP3—that can be solved sequentially. The proposed decomposition into SP1, SP2, and SP3 aligns naturally with the hierarchical control of the system, as illustrated in the resource management model. Each Slice Manager (SM) is responsible for intra-slice resource allocation (SP1), while the Slice Resource Orchestrator (SRO) manages inter-slice allocation (SP2) and global offloading decisions (SP3). This design enables distributed decision-making across the COIN–MEC infrastructure while maintaining coordination.

In deployment, DeepSets-S operates at both the slice and network levels. At the slice level, each SM hosts a local DeepSets-S instance that executes the SP1 model to allocate its wireless and computing resources among connected WDs. This local inference allows slices to adapt rapidly to short-term variations in workload and queue states without requiring centralized intervention. At the network level, the SRO runs the global DeepSets-S components for SP2 and SP3, aggregating summarized state information from all SMs.

Simulation results demonstrated the effectiveness of this design. DeepSets-S achieved tolerance-based accuracies above 95\% for SP1 and SP2, significantly improved multiclass offloading accuracy for SP3, and reduced computation time by 86.1\% compared to exact optimization solvers. At the system level, the proposed approach maintained system cost deviations within 6.1\% across COIN configurations and 3.8\% across varying slice numbers, confirming its ability to approximate optimal performance with substantially reduced overhead.

Practically, this framework offers a scalable foundation for AI-native orchestration in future 6G edge systems, combining distributed learning with centralized oversight. Looking ahead, several research directions arise. First, the present work assumes static user and channel conditions, extending the framework to dynamic or stochastic environments via reinforcement or continual learning would improve adaptability. Second, as generating optimal labels through solvers is costly, semi-supervised or self-supervised learning could reduce data dependency and improve generalization to unseen scenarios. Finally, evaluating DeepSets-S on real-time application workloads or public edge-computing datasets will be essential to assess scalability and deployment feasibility.

Beyond Metaverse environments, the proposed framework can be extended to other latency-sensitive domains such as autonomous systems, holographic media streaming, and immersive XR applications. Real-world implementation, however, will depend on the processing and memory capabilities of COIN nodes, which must be carefully analyzed to ensure efficient and stable on-device inference.

\section{Conclusion}

We have addressed the joint resource management and task offloading problem in a slice-enabled in-network edge system, where COIN and MEC nodes collaborate to support delay-sensitive and computation-intensive tasks. The problem was formulated as an MINLP model that jointly optimizes wireless and computing resource allocations with optimal slice selection to minimize overall system cost. As the formulated problem is NP-hard, we decomposed it into three tractable sub-problems, SP1 (intra-slice allocation), SP2 (inter-slice allocation), and SP3 (offloading decision).

To overcome the limitations of online optimization, we proposed a distributed hierarchical encoder-decoder framework (DeepSets-S), which learns to approximate optimal resource allocation and offloading policies from optimization solver-generated data. The model incorporates a shared encoder and task-specific decoders, ensuring permutation equivariance over variable-size WD sets. Once trained, DeepSets-S provides near-optimal inference with substantially reduced computational cost.


\bibliographystyle{ieeetr}
\bibliography{access}

\end{document}